\documentclass[conference]{IEEEtran}
\IEEEoverridecommandlockouts
\usepackage{cite}
\usepackage{amsmath,amssymb,amsfonts}
\usepackage{algorithmic}
\usepackage{amsthm}
\usepackage{graphicx}
\usepackage{textcomp}
\usepackage{xcolor}
\usepackage{subfig}
\usepackage{booktabs}
\usepackage[font=footnotesize]{caption}
\def\BibTeX{{\rm B\kern-.05em{\sc i\kern-.025em b}\kern-.08em
    T\kern-.1667em\lower.7ex\hbox{E}\kern-.125emX}}

\newtheorem{theorem}{Theorem}[section]

\usepackage{xspace}
\usepackage{bbm}

\newcommand{\cv}{{\bf c}}

\newcommand{\hv}{{\bf h}}

\newcommand{\wv}{{\bf w}}
\newcommand{\qv}{{\bf q}}
\newcommand{\xv}{{\bf x}}
\newcommand{\yv}{{\bf y}}
\newcommand{\zv}{{\bf z}}

\def\textiid{i.i.d.\@\xspace}
\newcommand\iid{\ifmmode\text{ i.i.d. } \else \textiid \fi}

\begin{document}

\title{Precoding-Oriented CSI Feedback Design
with Mutual Information Regularized VQ-VAE}

\author{\IEEEauthorblockN{Xi Chen}
\IEEEauthorblockA{\textit{Department of ECE, Rutgers University} \\
Piscataway, NJ, USA \\
xi.chen15@rutgers.edu}
\and
\IEEEauthorblockN{Homa Esfahanizadeh}
\IEEEauthorblockA{\textit{Nokia Bell Labs} \\
Murray Hill, NJ, USA \\
homa.esfahanizadeh@nokia-bell-labs.com}
\and
\IEEEauthorblockN{Foad Sohrabi}
\IEEEauthorblockA{\textit{Nokia Bell Labs} \\
Murray Hill, NJ, USA \\
foad.sohrabi@nokia-bell-labs.com}
}

\maketitle

\begin{abstract}
Efficient channel state information (CSI) compression at the user equipment plays a key role in enabling accurate channel reconstruction and precoder design in massive multiple-input multiple-output systems. A key challenge lies in balancing the CSI feedback overhead with the achievable downlink rate, i.e., maximizing the utility of limited feedback to maintain high system performance. In this work, we propose a precoding-oriented CSI feedback framework based on a vector quantized variational autoencoder, augmented with an information-theoretic regularization. To achieve this, we introduce a differentiable mutual information lower-bound estimator as a training regularizer to promote effective utilization of the learned codebook under a fixed feedback budget. Numerical results demonstrate that the proposed method achieves rates comparable to variable-length neural compression schemes, while operating with fixed-length feedback. Furthermore, the learned codewords exhibit significantly more uniform usage and capture interpretable structures that are strongly correlated with the underlying channel state information.
\end{abstract}

\begin{IEEEkeywords}
Channel state information (CSI), multiple-input multiple-output (MIMO), mutual information (MI), vector quantized variational autoencoder (VQ-VAE)
\end{IEEEkeywords}

\section{Introduction}
One of the most promising technologies to meet the stringent key performance indicators of future wireless communications is massive multiple-input multiple-output (MIMO) communication. To exploit the full potential of the massive MIMO systems, accurate channel state information (CSI) must be obtained at the base station (BS). One conventional way to obtain the channel quality of a user, applicable in both frequency-division duplex (FDD) and time-division duplex (TDD) operation modes, involves performing downlink channel training. In this method, the BS first broadcasts pilots via downlink channels. Each user equipment (UE) then estimates its own CSI based on the received pilots. Finally, each UE feeds back the estimated CSI to the BS via dedicated uplink channels. The CSI feedback from the UE can then be used for various purposes, such as precoding.

In massive MIMO downlink training, CSI compression at the UE is critical for accurate channel reconstruction at the BS and for reducing the uplink feedback overhead. Existing approaches range from analytically motivated compression schemes based on sparsity or low-rank channel models \cite{love2008overview,rao2014distributed,schniter2014channel,gao2015spatially} to data-driven autoencoder architectures that jointly optimize compression and reconstruction \cite{guo2020deep,guo2022overview}. Under strict feedback bandwidth constraints, effective CSI compression must preserve the maximum possible information from the UE’s received pilot signals by fully exploiting the available feedback bits to represent this information. Deep learning–based precoding-oriented CSI feedback methods have been shown to improve the final achievable sum rate \cite{sohrabi2021deep}. However, existing approaches often fail to fully utilize the available feedback bits, as the learned codewords are non-uniformly used, leading to unnecessary information loss and suboptimal performance. While neural entropy coding has been proposed to address this issue by enabling variable-length feedback \cite{carpi2023precoding}, it introduces additional decoding complexity to the system, which complicates practical system deployment.

In this work, we propose a framework based on vector quantized variational autoencoder (VQ-VAE) \cite{van2017neural} for precoding-oriented CSI feedback design. In the downlink training phase, VQ-VAE maps the continuous encoder input, i.e., noisy pilot observations at the UE, to a discrete latent representation via an explicitly learned codebook, where each encoder output is quantized to its nearest codeword in Euclidean distance. The corresponding codeword indices are sent back to the BS, which then decodes the selected codewords to construct linear precoders that maximize the UEs’ sum achievable rate.

Under a fixed feedback bandwidth constraint, we aim to maximize the effective utilization of the learned codebook. To achieve this, we derive a novel mutual information (MI) lower bound estimator to measure the dependence between the UE’s noisy observations and the codewords indexed in the feedback. Additionally, we introduce a differentiable relaxation that that facilitates the integration of the MI estimator as a regularization term in end-to-end VQ-VAE training. Maximizing this MI encourages high-entropy (i.e, more uniform) codeword usage due to the deterministic nature of VQ-VAE quantization.
Compared to existing baseline approaches, our method improves codeword utilization uniformity and achieves higher sum achievable rates under fixed-length feedback in massive MIMO systems. Furthermore, we provide interpretability insights, showing that the learned codewords are strongly correlated with CSI.

\section{System Model}
We consider the downlink of an FDD massive MIMO system, where a BS equipped with $M$ transmit antennas serves $K$ single-antenna UEs, with $K \leq M$.
The BS employs linear precoding, and the transmitted signal can be written as $\mathbf{x} = \sum_{k=1}^{K} \mathbf{v}_k s_k = \mathbf{V}\mathbf{s}$, where $\mathbf{v}_k \in \mathbb{C}^{M}$ is the precoding vector for the $k$-th UE, forming the $k$-th column of the precoding matrix $\mathbf{V} \in \mathbb{C}^{M \times K}$. The transmitted symbols are normalized such that
$\mathbb{E}[\mathbf{s}\mathbf{s}^H] = \mathbf{I}_K$, and the total power constraint requires $\mathrm{Tr}(\mathbf{V}\mathbf{V}^H) \le P$. Assuming a narrowband block-fading channel model, the received signal at the $k$-th UE during the data transmission phase is $y_k = \mathbf{h}_k^H \mathbf{v}_k s_k + \sum_{j \ne k} \mathbf{h}_k^H \mathbf{v}_j s_j + w_k$, where $\mathbf{h}_k \in \mathbb{C}^{M}$ denotes the downlink channel vector from the BS to the $k$-th UE, and $w_k \sim \mathcal{CN}(0,\sigma_w^2)$ is additive white Gaussian noise. The achievable rate (bps/Hz) of UE $k$ is
\begin{equation} \label{eq:rate}
    R_k = \log_2 \!\left( 1 + 
    \frac{|\mathbf{h}_k^H \mathbf{v}_k|^2}
    {\sum_{j \ne k} |\mathbf{h}_k^H \mathbf{v}_j|^2 + \sigma_w^2} \right).
\end{equation}
The system sum rate of $K$ UEs is defined as $R = \sum_{k=1}^{K} R_k$. 

To design the precoding matrix $\mathbf{V}$, the BS requires access to instantaneous CSI. We assume that neither the BS nor the UEs have prior knowledge of the channel realizations, and CSI is acquired via downlink training and feedback. In the downlink training, the BS transmits pilot symbols $\widetilde{\mathbf{X}} \in \mathbb{C}^{M \times L}$, where $L$ is the training length. The received pilot signal at the $k$-th UE is
\begin{equation}
\yv_k = \hv_k^{H} \widetilde{\mathbf{X}} + \wv_k,
\end{equation}
where $\wv_k \sim \mathcal{CN}(\mathbf{0}, \sigma_w^2 \mathbf{I}_L)$. Each pilot vector $\widetilde{\mathbf{x}}_\ell$ satisfies the power constraint $\|\widetilde{\mathbf{x}}_\ell\|_2^2 \le P$. When $\mathbf{y}_k$ is received by the \mbox{$k$-th} UE,  it is mapped to $B$ feedback bits by a function as
\begin{align}
    \mathbf{q}_k = \mathcal{F}_k(\mathbf{y}_k), 
    \quad \mathcal{F}_k : \mathbb{C}^{1 \times L} \rightarrow \{ 0,1 \}^B.
\end{align}
The BS collects the feedback from all $K$ UEs, denoted by
$\mathbf{q} = [\mathbf{q}_1^T,\ldots,\mathbf{q}_K^T]^T$, and designs the precoding matrix using the function $\mathcal{P}$ as
\begin{align}
    \mathbf{V} = \mathcal{P}(\mathbf{q}),
    \quad \mathcal{P} : \{0, 1\}^{KB} \rightarrow \mathbb{C}^{M \times K}.
\end{align}
Given the above model, the sum-rate maximization problem under limited feedback can be formulated as
\begin{align}
    \max_{\widetilde{\mathbf{X}}, \{\mathcal{F}_k\}^K, \mathcal{P}}
    \sum_{k=1}^{K} 
    R_k , \:
    \text{s.t.} \:
     \mathrm{Tr}\left(\mathbf{V}\mathbf{V}^H\right) \le P, \: \|\widetilde{\mathbf{x}}_\ell\|_2^2 \le P, \forall \ell.
\end{align}
The formulated optimization problem above involves jointly \mbox{designing} the downlink training pilots, UE-side CSI estimation and compression, and BS-side precoding based on limited feedback. The strong coupling between channel estimation, feedback compression, and precoder design makes analytical optimization intractable, which motivates the use of data-driven learning-based approaches.



\section{Learning-based Precoding-Oriented CSI Feedback Design}
In this section, we introduce the end-to-end training \mbox{framework} for precoding-oriented CSI feedback design. On the UE side, to obtain a better compressed representation of the CSI, a non-linear transform is applied to the received pilot signals. The encoder $f_{\psi_k}^{\text{enc}}(\cdot)$ is parameterized by $\psi_k$ to encode $\yv_k$ as $\zv_k = f^{\text{enc}}_{\psi_k}(\yv_k)$, where $\zv_k \in \mathbb{R}^{D}$ is the continuous latent representation and $D$ is the dimension of the latent representation.
Based on the bandwidth of the feedback link, the encoded feedback containing the estimated CSI is quantized using a certain number of bits, which is then sent back to BS. In a more specific model design, $\zv_k$ is quantized for downlink channel UE feedback using a quantization function $Q(\cdot)$. The BS collects the feedback from $K$ UEs, and uses it to design the precoder through a decoder $f_{\phi}^{\text{dec}}(\cdot)$ parameterized by $\phi$ as
\begin{align}
    V = f^{\text{dec}}_{\phi}(\{Q(\zv_k)\}^K_{k=1}) \in \mathbb{C}^{M \times K}.
\end{align}
Given that the goal of precoding design is to maximize the sum achievable rate, this can be structured as goal-oriented code design where the objective is not to recover the source. Instead, the focus is on designing the precoder matrix for end-to-end training of the parameters $(\{\psi_k\}^K_{k=1}, \phi, \widetilde{X})$ in the model.

Since the quantization step is non-differentiable in the end-to-end training, one common approach is to use straight-through estimation (STE) \cite{bengio2013estimatingpropagatinggradientsstochastic} for gradient backpropagation.
However, quantizing the latent logits induces coarse binary partitions of the latent space and provides limited control over codeword utilization. As a result, the trained model often exhibits highly uneven codeword usage in UE feedback, leading to inefficient bandwidth utilization and suboptimal achievable rates. In the next section, we propose an alternative framework to address this issue.





\begin{figure*}[htb]
\centering
\includegraphics[width=0.76\textwidth]{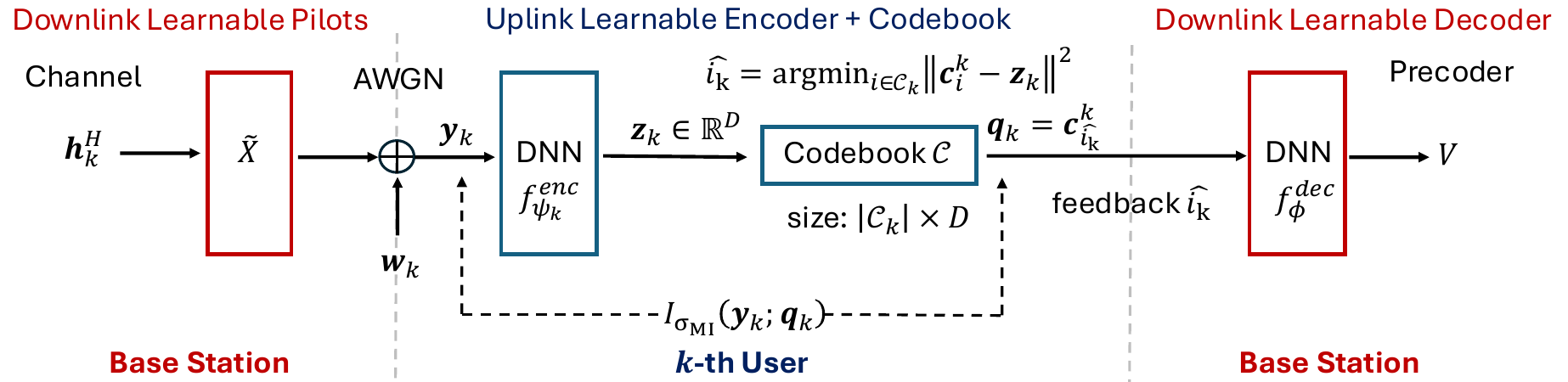}
\caption{The structure of VQ-VAE with MI regularizer illustrated for a single UE precoder-oriented CSI feedback design.}
\label{fig:VQVAE_single}
\vspace{-0.2in}
\end{figure*}

\section{Precoding-oriented CSI Feedback Design with Mutual Information Regularized VQ-VAE}
To tackle the problem of inefficient bandwidth utilization in the existing learning-based precoding-oriented CSI feedback design works, we propose using VQ-VAE for the quantization of UE feedback. This is because VQ-VAE employs an explicitly learned codebook that enables more expressive and structured quantization. 
This section presents a VQ-VAE–based framework in details for precoder-oriented CSI feedback. Furthermore, under a fixed feedback bandwidth constraint, we introduce a differentiable MI lower bound as a regularization term to promote more efficient utilization of the feedback bits.
\subsection{Vector Quantized Variational Autoencoder}
The VQ-VAE method is known to provide the discrete latent representation of the encoder input by first mapping the input source to a latent representation using non-linear neural networks. 
Quantization is achieved by searching the codebook for the codeword closest to the continuous representation in terms of Euclidean distance \cite{van2017neural}. The indices of the closest codewords are used as feedback bits from the UE side, and the codebook is learned simultaneously and shared between the BS and UE. In precoder-oriented CSI feedback design, VQ-VAE operates as follows: Given the downlink signal for each UE $\yv_k$, a separate encoder maps the signal to the latent representation as 
\begin{align}
    \zv_k = f^{\mathrm{enc}}_{\psi_k}(\yv_k) \in \mathbb{R}^{D}.
\end{align}
The codebook $\mathcal{C}_k$ for the $k$-th UE is defined as a finite set 
\begin{align}
    \mathcal{C}_k = \{(i,\, \cv^k_{i})\}^{|\mathcal{C}_k|}_{i=1},
\end{align}
where $i, \cv_i$ are the index and its corresponding codeword, respectively. The cardinality of the set is $|\mathcal{C}_k|$, and each \mbox{codeword} has dimension $D$. Quantization is performed by selecting the codeword closest to $\zv_k$ in terms of $\ell^2$ distance, with the detailed steps as
\begin{align}
    \hat{i}_k = \arg\min_{i \in \mathcal{C}_k} \|\cv^k_i - \zv_k\|_2^2, \qquad
    \qv_k = \cv^k_{\hat{i}_k}.
\end{align}
For simplicity, we denote the quantization step as
\begin{align}
    \qv_k = Q(\zv_k, \mathcal{C}_k).
\end{align}
After quantization, the indices of the selected codewords are sent to the BS in uplink feedback using a fixed-length code of $\log_2|\mathcal{C}_k|$ bits. The quantized vector $\qv_k$ corresponding to $\zv_k$ is retrieved from the shared codebook.  Feedback from all UEs is used at the BS to form the precoder as 
\begin{align}
    \mathbf{V} = f^{\mathrm{dec}}_{\phi}(\{\qv_k\}^K_{k=1}) \in \mathbb{C}^{M \times K}.
\end{align}

To enhance expressiveness and ensure feasible codebook implementation, our proposed solution learns multiple codebooks for a single UE.
 For example, allocating $20$ bits would require a codebook with $\log_2|\mathcal{C}|=20$, resulting in $2^{20}$ codewords, which is impractical. Instead, we use two separate codebooks, each represented with 10 bits, allowing the transmission of more bits while maintaining a manageable codebook size. Since the learning dynamics of codebooks differ from those of neural encoders and decoders, we adopt a separate learning rate for the codebook. This approach has been empirically shown to be crucial for stable training and effective codeword utilization.


\subsection{Differentiable MI Lower Bound}
For explicitly learned codebooks, codebook collapse is a well-known issue \cite{takida2022sqvaevariationalbayesdiscrete}, where only a small subset of codewords is frequently selected while others are rarely or never used, leading to inefficient utilization of the available feedback resources. To mitigate this problem, we propose to maximize the mutual information between the received signal at the UE, $\yv$, and the selected codeword, $\qv$. Under a fixed codebook size and limited feedback budget, maximizing this mutual information is equivalent to increasing the entropy of the codeword usage distribution, which promotes more uniform codeword utilization and improves the efficiency of the learned representations. This is because the quantized latent representation is a deterministic function of the original representation. Thus, this MI maximization is equivalent to maximizing entropy of the quantized encoded data and equivalently enforcing it to have close to uniform distribution as possible.

We first introduce a novel MI lower bound, which requires no training unlike the deep learning-based MI estimator, e.g., MINE \cite{belghazi2018mutual} and ReMINE \cite{choi2021regularized}.
\begin{theorem}\label{theorem:MI_LB_1}
    Assuming $Z=T(X)$, where $T:\mathcal{X}\rightarrow\mathcal{Z}$ is a deterministic function, we have
    \begin{equation*}
        \mathcal{I}(X,Z)\geq-\log \mathbb{E}_{P_{{Z}}P_{{Z}}} [\delta(z-z')],
    \end{equation*}
    where $z$, $z'$ are two r.v.s. independently drawn from $P_Z$. Here, 
    $\delta(\cdot)=1$ if its argument is zero, and $\delta(\cdot)=0$ otherwise.
\end{theorem}
\begin{proof}
We have
\[
I(X;Z) = H(Z) - H(Z|X) = H(Z),
\]
based on the definition of MI and the fact that $T$ is deterministic, which implies
$H(Z|X) = 0$. By Jensen's inequality,
\[
H(Z) 
= \mathbb{E}_{p_Z}\!\left[- \log P_Z(z)\right]
\ge - \log \, \mathbb{E}_{p_Z}\!\left[P_Z(z)\right].
\]
Further,
\begin{align*}
- \log \mathbb{E}_{p_Z}\!\left[P_Z(z)\right]
&= - \log \, \mathbb{E}_{p_Z} \left[ \sum_{z'} \delta(z - z') P_Z(z') \right] \\
&= - \log \mathbb{E}_{p_Z} \mathbb{E}_{p_{Z'}} 
      \!\left[ \delta(z - z') \right] \\
&= - \log \mathbb{E}_{p_Z p_{Z'}} 
      \!\left[ \delta(z - z') \right].
\end{align*}
Therefore,
\[
I(X;Z) \ge 
- \log \, 
\mathbb{E}_{p_{Z} p_{Z'}}\!\left[ \delta(z - z') \right].
\qedhere
\]
We note that the approximation error comes from Jensen's 
inequality, which can be tight under certain conditions.
\end{proof}
Given the non-differentiability of the indicator function $\delta(\cdot)$, it is challenging to plug the MI term into an end-to-end deep-learning model, as well as in VQ-VAE. Therefore, we relax the MI estimator by replacing the indicator function with a Gaussian kernel, as follows
\begin{align}
    I(X;Z) \ge 
    - \log \, 
    \mathbb{E}_{p_Z p_{Z'}} 
    \!\left[
    \exp\!\left( 
    - \frac{\|z - z'\|^2}{2 \sigma_{\text{MI}}^2}
    \right)
    \right],
\end{align}
where $\sigma_{\text{MI}}$ is the variance, which controls the approximation accuracy. When $\sigma_{\text{MI}}$ approaches zero, the kernel approximator approaches the delta function, but also affects the training stability. The expectation over $p_Z p_{Z'}$ is estimated by sampling pairs with independent elements $(z, z')$ from the latent space and averaging, given that the latent distribution is unknown.

\subsection{MI Empowered VQ-VAE for CSI Feedback Design}
Note that the codebook(s) may be shared across UEs or assigned individually, with each UE employing one or multiple codebooks. The codebook(s) are known to both the UEs and the BS; thus, the uplink feedback from each UE consists of $N \times \log_2 |\mathcal{C}|$ bits, where $N$ is the number of codebooks per UE, assuming all codebooks have the same number of codewords. The BS aggregates the feedback from all UEs and uses the decoder to construct the precoder matrix $\mathbf{V} \in \mathbb{C}^{M \times K}$. This precoder is then incorporated into the end-to-end training objective of VQ-VAE. The objective function to be minimized comprises three terms:
\begin{align}
    \mathcal{L}(\psi,\phi,\mathcal{C})
    = \sum_{k=1}^{K} - R_k 
    - \gamma\, I\left( \yv_k; \qv_k \right)
    + C L_{k},
\end{align}
where $R_k$ is defined in \eqref{eq:rate}. The estimated mutual information between $\yv_k$ and $\qv_k$, where $\qv_k = Q(f^{\mathrm{enc}}_{\psi_k}(\yv_k), \mathcal{C}_k)$, is
\begin{align*}
    I\left( \yv_k ; \qv_k \right)
    = - \log 
    \left(
    \frac{1}{\kappa S}\sum^{\kappa S}_{b \neq b'} 
    \exp \Bigg(
    - \frac{ \| \qv^{(b)}_k - \qv^{(b')}_k \|^2}{2\sigma_{\text{MI}}^2}
    \Bigg)
    \right).
\end{align*}
Here, $b$, $b'$ denote the different indices in the training batch, $S$ is the batch size, the Monte-Carlo approach $\frac{1}{\kappa S}\sum^{\kappa S}_{b\neq b'}$ is applied batch-wise on pairs of samples in latent space. To improve the estimation accuracy, the sampling size is $\kappa$ times of the batch size. The commitment loss term $CL_k$ constrains the learned codebook to be a good quantizer
\begin{align}
    CL_k 
    = \beta \left\| \mathrm{sg}(\qv_k) - \zv_k \right\|^2 
     + \left\| \qv_k - \mathrm{sg}(\zv_k) \right\|^2,
\end{align}
where $\mathrm{sg}(\cdot)$ denotes the stop-gradient operator. $\gamma, \beta$ are \mbox{hyperparameters} for balancing the terms in loss function. In 
estimation of $I(\yv_k; \qv_k)$, the STE is used as
\begin{align}
    \qv_k = \zv_k + \mathrm{sg}(\qv_k - \zv_k),
\end{align}
which handles the non-differentiability of the quantization operation. The pipeline is shown in Fig.~\ref{fig:VQVAE_single}. The constrained optimization problem of the downlink channel training is
\begin{align}
    \min_{\widetilde{\mathbf{X}}, \psi, \phi, \mathcal{C}}
    \quad 
    \sum_{k=1}^{K} - R_k 
    - \gamma I\left( \yv_k; \qv_k \right)
    + C L_k,
\end{align}
\[
\text{s.t.} \quad  \forall l \:\:\:  \|\xv_{l}\|^{2} \le P, \quad  \operatorname{Tr}\!\left( \mathbf{V}\mathbf{V}^{H} \right) \le P.
\]
where the optimization variables are the trainable pilots $\widetilde{\mathbf{X}}$ and VQ-VAE parameters $(\psi, \phi, \mathcal{C})$. For simplicity, we use $(\psi, \phi, \mathcal{C})$ to collectively denote all VQ-VAE parameters, although in implementation each UE employs a separate encoder and may use multiple codebooks.

\section{Numerical results}
We compare the proposed method with both oracle CSI at the transmitter (CSIT)-based precoding schemes and learning-based CSI feedback approaches. In particular, we include two classical linear precoding methods: maximum ratio transmission (MRT) and zero-forcing (ZF), which assume access to full CSIT. These methods serve as standard linear-precoding benchmarks under full CSIT. While MRT maximizes the desired signal power and is optimal in interference-free regimes, ZF mitigates multi-user interference, achieving near-optimal performance in multi-user settings. We also compare against two representative deep learning methods: (1) precoding-oriented DNNs with sign function as UE feedback quantization \cite{sohrabi2021deep}; (2) precoding-oriented DNNs with neural entropy coding on the UE feedback \cite{carpi2023precoding}.

\subsection{Channel Model}
We consider an FDD massive MIMO system operating in a mmWave propagation environment where the number of scatters is limited. Therefore, the channel of the $k$-th UE is modeled with $L_p$ propagation paths as \cite{pi2011introduction, sohrabi2016hybrid}
\begin{align} \label{eq:channel_model}
    \mathbf{h}_k 
    = \frac{1}{\sqrt{L_p}}
    \sum_{l=1}^{L_p} \alpha_{l,k}\, \mathbf{a}_t(\theta_{l,k}),
\end{align}
where $\alpha_{l,k}$ is the complex gain of the $l$-th path between the BS and UE (capturing phase shift and attenuation), $\theta_{l,k}$ is the angle of departure (AoD), and $\mathbf{a}_t(\cdot)$ is the transmit array response vector. To model a uniform linear array with $M$ antenna elements, the transmit 
array response vector is
\begin{align*}
    \mathbf{a}_t(\theta)
    =
    \begin{bmatrix}
    1,\;
    e^{j\frac{2\pi}{\lambda} d \sin(\theta)},\;
    \ldots,\;
    e^{j\frac{2\pi}{\lambda} d (M-1)\sin(\theta)}
    \end{bmatrix}^{T},
\end{align*}
where $\lambda$ is the wavelength and $d$ is the antenna spacing in the linear array. The path fading is modeled as $\alpha_{l,k} \sim \mathcal{CN} \left(0,\sigma^{2}\right)$. 

\begin{figure}[h]
    \centering
    \includegraphics[width=0.75\columnwidth]{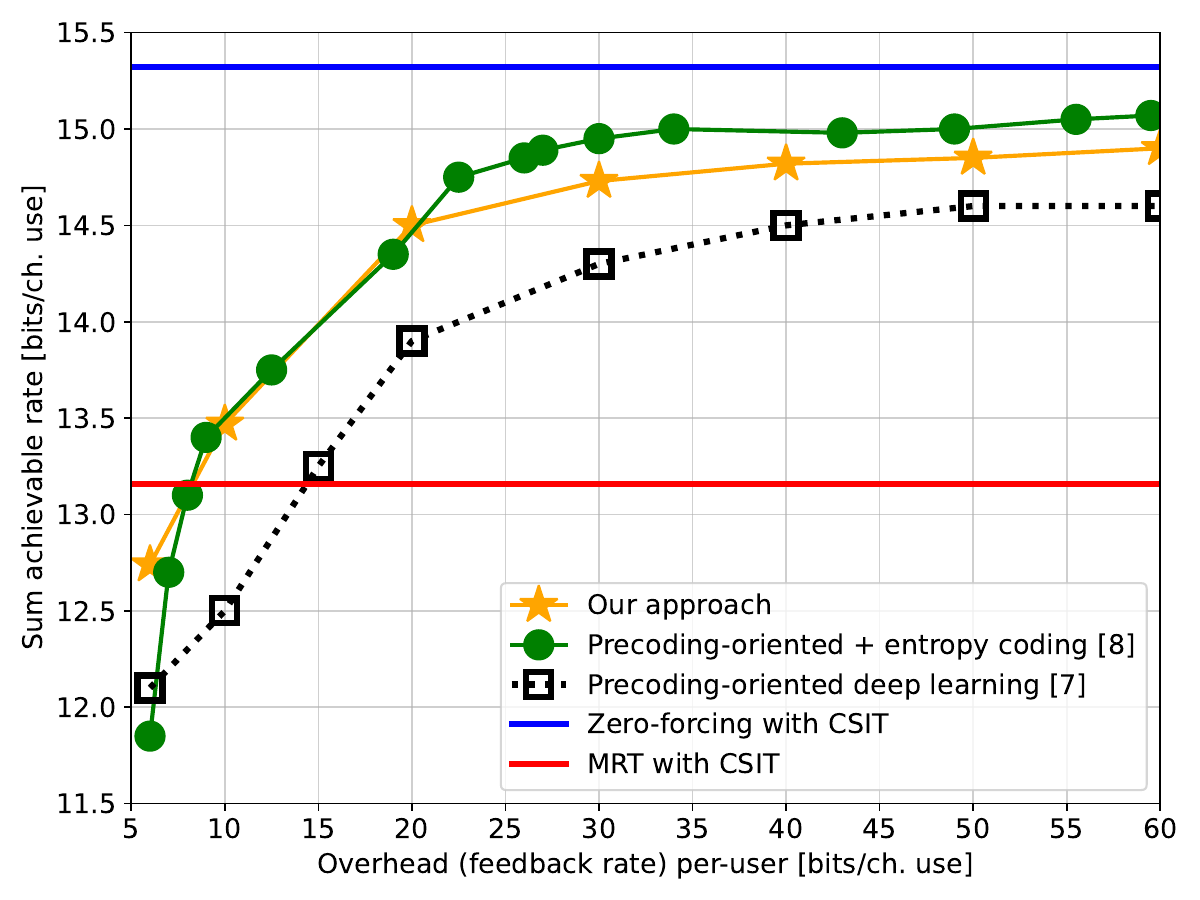}
    \caption{Performance comparison in terms of sum achievable rate.}
    \label{fig:rate_comparison}
    \vspace{-0.1in}
\end{figure}

\subsection{Sum Achievable Rate v.s. Feedback Bits Performance}
We train the proposed VQ-VAE for CSI feedback using the channel model in \eqref{eq:channel_model}. All channel and training hyperparameters are summarized in Table~\ref{tab:exp_setup}. Channel realizations are freshly generated in each mini-batch according to the specified distributions. The number of batches per epoch is $20$, and each batch contains $10{,}000$ samples. Model parameters are optimized using the Adam optimizer. We evaluated our method and the baselines based on the sum achievable rate, and plotted against the overhead feedback rate per UE in Fig \ref{fig:rate_comparison}. It can be seen that our approach outperforms the deep-learning-based method \cite{sohrabi2021deep} and achieves comparable performance to the neural entropy coding approach \cite{carpi2023precoding}, especially when the overhead is less than $10$ bits, while being more practical.

\begin{figure}[h]
    \centering
    \includegraphics[width=0.75\columnwidth]{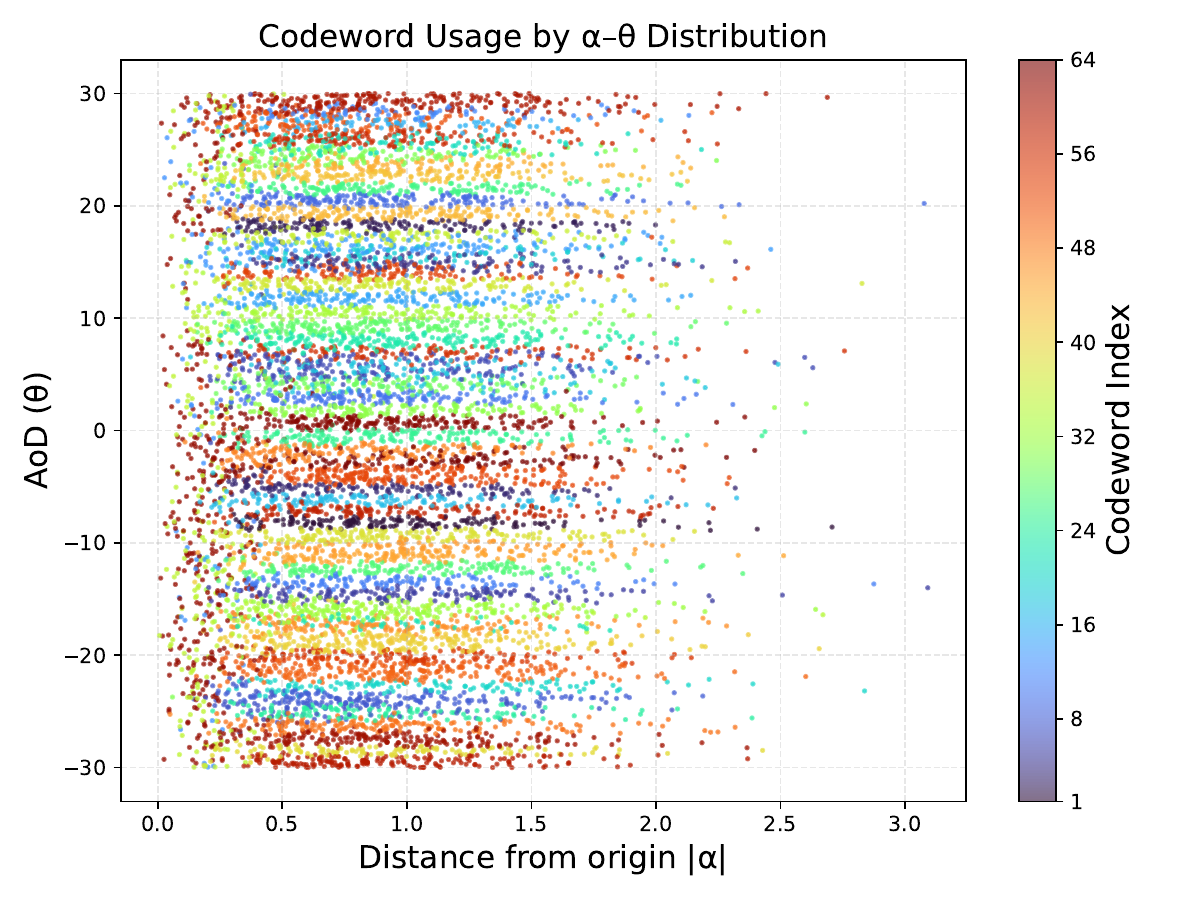}
    \caption{Channel information and its assigned codeword.}
    \label{fig:codeword_vis}
    \vspace{-0.25in}
\end{figure}

\subsection{Channel interpretability of learned codeword}
Since the codebook is explicitly learned in our framework, we further evaluate and interpret the relationship between the learned codewords and the underlying channel state information. To make the visualization simple, here we focus on a single-path scenario where the channel can be represented by only a few parameters. We generate $10{,}000$ channel realizations using a single-path channel model and visualize how the CSI of an individual UE is mapped to the learned codewords. Under a $6$-bit feedback constraint per UE, with the channel model parameterized by the AoD and the path fading coefficient $\alpha$, Fig.~\ref{fig:codeword_vis} illustrates the association between CSI and learned codewords. We observe that similar AoD results in assignment to similar codewords, indicating that $\theta$, the dominant factor characterizing the channel, plays a more significant role in the learned codebook structure. This observation is consistent with domain knowledge in precoder design, where the channel direction plays a more critical role than its magnitude. Moreover, we observe that a few codewords are predominantly associated with weak channel realizations. This suggests that when the received SNR is low, the encoder tends to map such channels to a small subset of codewords, reflecting the limited impact that these UEs have on the overall precoding design since their achievable rates are inherently constrained.

\begin{figure}[h]
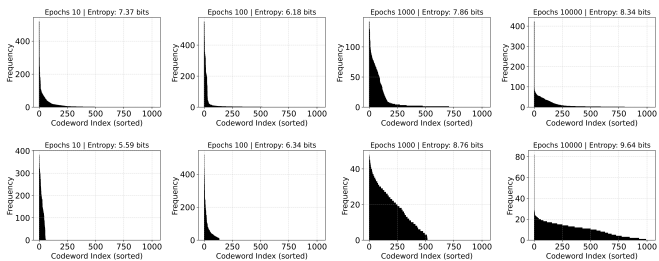

    \centering
    \includegraphics[width=1.0\columnwidth]{figures/histogram_U1_binary_quantization.png}\\
    \includegraphics[width=1.0\columnwidth]{figures/histogram_U1_VQVAE_MI.png}
    \caption{The number of feedback bits per UE is fixed to $10$. Histograms of codeword usage across different training epochs: (Top) DNN-based method in \cite{sohrabi2021deep}; (Bottom) Our proposed method exhibits a more uniform distribution.}
    \label{fig:codeword_usage}
    \vspace{-0.2in}
\end{figure}

\subsection{Codebook Usage with Allocated Bits}
We obtain the histogram of codeword usage throughout training and compare our approach with the baseline in Fig.~\ref{fig:codeword_usage}. As training progresses, the proposed method yields a more uniformly distributed codeword usage than the approach in \cite{sohrabi2021deep}, effectively mitigating the codebook collapse issue.


\section{Conclusion}
We introduced an MI regularized VQ-VAE framework for precoding-oriented CSI feedback in massive MIMO systems. By explicitly promoting uniform and informative codeword utilization under a feedback bandwidth constraint, the proposed approach mitigates codebook collapse and achieves improved sum achievable rates compared with existing methods. In addition, the learned codebook exhibits meaningful structure correlated with key channel parameters, providing \mbox{interpretability} into the learned CSI representations.

\begin{table}[t]
\centering
\scriptsize
\caption{Simulation and training configuration}
\label{tab:exp_setup}
\resizebox{\columnwidth}{!}{
\begin{tabular}{ll}
\toprule
Number of antennas $M$, UEs $K$, paths $L_p$, pilots / UE $L$ & $64$, $2$, $2$, $8$ \\
SNR & $10$ dB \\
AoD $\theta_{l,k}$ & $\sim \mathcal{U}(-\pi/6,\pi/6)$ \\
Path gain $\alpha_{l,k}$ & $\sim \mathcal{CN}(0,1)$ \\ \midrule
Codeword dimension $D$ & $256$ \\
Codebook learning rate, Other parameters learning rate & $1\times 10^{-2}$, $1\times 10^{-4}$ \\
Commitment loss weight $\beta$ & $0.1$ \\
MI relaxation parameter $\sigma_{\text{MI}}$, loss weight $\gamma$, sampling parameter $\kappa$ & $0.05$, $1.0$, $16$ \\
\bottomrule
\end{tabular}
}
\vspace{-0.2in}
\end{table}

\bibliographystyle{IEEEtran}
\bibliography{refs}

@article{sohrabi2021deep,
  title={Deep learning for distributed channel feedback and multiuser precoding in {FDD} massive {MIMO}},
  author={Foad Sohrabi and Kareem M Attiah and Wei Yu},
  journal={IEEE Trans. Wireless Commun.},
  volume={20},
  number={7},
  pages={4044--4057},
  year={2021},
  publisher={IEEE}
}

@inproceedings{carpi2023precoding,
  author    = {Fabrizio Carpi and Sivarama Venkatesan and Jinfeng Du and Harish Viswanathan and Siddharth Garg and Elza Erkip},
  title     = {Precoding-oriented massive {MIMO} {CSI} feedback design},
  booktitle = {Proc. IEEE Int. Conf. Commun. (ICC)},
  address   = {Rome, Italy},
  month     = may,
  year      = {2023},
  pages     = {4973--4978}
}

@inproceedings{belghazi2018mutual,
  title={Mutual information neural estimation},
  author={Belghazi, Mohamed Ishmael and Baratin, Aristide and Rajeshwar, Sai and Ozair, Sherjil and Bengio, Yoshua and Courville, Aaron and Hjelm, Devon},
  booktitle={ICML},
  pages={531--540},
  year={2018},
  organization={PMLR}
}

@misc{
choi2021regularized,
title={Regularized Mutual Information Neural Estimation},
author={Kwanghee Choi and Siyeong Lee},
year={2021},
url={https://openreview.net/forum?id=Lvb2BKqL49a}
}

@misc{takida2022sqvaevariationalbayesdiscrete,
      title={{SQ-VAE}: Variational Bayes on Discrete Representation with Self-annealed Stochastic Quantization}, 
      author={Yuhta Takida and others},
      year={2022},
      eprint={2205.07547},
      archivePrefix={arXiv},
      primaryClass={cs.LG},
      url={https://arxiv.org/abs/2205.07547}, 
}

@article{pi2011introduction,
  title={An introduction to millimeter-wave mobile broadband systems},
  author={Pi, Zhouyue and Khan, Farooq},
  journal={IEEE Commun. Mag.},
  volume={49},
  number={6},
  pages={101--107},
  year={2011},
  publisher={IEEE}
}

@article{sohrabi2016hybrid,
  title={Hybrid digital and analog beamforming design for large-scale antenna arrays},
  author={Sohrabi, Foad and Yu, Wei},
  journal={IEEE J. Sel. Topics Signal Process.},
  volume={10},
  number={3},
  pages={501--513},
  year={2016},
  publisher={IEEE}
}

@article{love2008overview,
  title={An overview of limited feedback in wireless communication systems},
  author={Love, David J and others},
  journal={IEEE J. Sel. Areas Commun.},
  volume={26},
  number={8},
  pages={1341--1365},
  year={2008},
  publisher={IEEE}
}

@article{rao2014distributed,
  title={Distributed compressive {CSIT} estimation and feedback for {FDD} multi-user massive {MIMO} systems},
  author={Rao, Xiongbin and Lau, Vincent KN},
  journal={IEEE Trans. Signal Process.},
  volume={62},
  number={12},
  pages={3261--3271},
  year={2014},
  publisher={IEEE}
}

@article{gao2015spatially,
  title={Spatially common sparsity based adaptive channel estimation and feedback for {FDD} massive {MIMO}},
  author={Gao, Zhen and Dai, Linglong and Wang, Zhaocheng and Chen, Sheng},
  journal={IEEE Trans. Signal Process.},
  volume={63},
  number={23},
  pages={6169--6183},
  year={2015},
  publisher={IEEE}
}

@inproceedings{schniter2014channel,
  author    = {Philip Schniter and Akbar Sayeed},
  title     = {Channel estimation and precoder design for millimeter-wave communications: The sparse way},
  booktitle = {Proc. Asilomar Conf. Signals, Systems, and Computers},
  address   = {Pacific Grove, CA, USA},
  month     = nov,
  year      = {2014},
  pages     = {273--277}
}

@article{guo2020deep,
  title={Deep learning-based {CSI} feedback for beamforming in single-and multi-cell massive {MIMO} systems},
  author={Guo, Jiajia and Wen, Chao-Kai and Jin, Shi},
  journal={IEEE J. Sel. Areas Commun.},
  volume={39},
  number={7},
  pages={1872--1884},
  year={2020},
  publisher={IEEE}
}

@article{guo2022overview,
  title={Overview of deep learning-based {CSI} feedback in massive {MIMO} systems},
  author={Guo, Jiajia and Wen, Chao-Kai and Jin, Shi and Li, Geoffrey Ye},
  journal={IEEE Transactions on Communications},
  volume={70},
  number={12},
  pages={8017--8045},
  year={2022},
  publisher={IEEE}
}

@article{van2017neural,
  title={Neural discrete representation learning},
  author={Van Den Oord, Aaron and others},
  journal={NeurIPS},
  volume={30},
  year={2017}
}

@misc{bengio2013estimatingpropagatinggradientsstochastic,
      title={Estimating or Propagating Gradients Through Stochastic Neurons for Conditional Computation}, 
      author={Yoshua Bengio and Nicholas Léonard and Aaron Courville},
      year={2013},
      eprint={1308.3432},
      archivePrefix={arXiv},
      primaryClass={cs.LG},
      url={https://arxiv.org/abs/1308.3432}, 
}

\end{document}